\documentclass[12pt]{amsart}

\usepackage[matrix,arrow,curve,frame]{xy}
\usepackage{amsmath,amsthm,amssymb,enumerate}
\usepackage{latexsym}
\usepackage{amscd}
\usepackage{euscript}

\usepackage{fullpage}

\usepackage[all]{xypic}

\newtheorem{theorem}{Theorem}[section]
\newtheorem{lemma}[theorem]{Lemma}
\newtheorem{corollary}[theorem]{Corollary}

\theoremstyle{definition}

\theoremstyle{remark}
\newtheorem{remark}[theorem]{Remark}

\numberwithin{equation}{section}
\newcommand{\beq}{\begin{equation}}
\newcommand{\eeq}{\end{equation}}

\newcommand{\TT}{\mathbb{T}}
\newcommand{\ZZ}{\mathbb{Z}}

\newcommand\bbT{\mathbb T}

\DeclareMathOperator{\Tr}{Tr}

\renewcommand{\cL}{\mathcal{L}}
\newcommand{\cG}{\mathcal{G}}

\newcommand{\gH}{\mathfrak{H}}
\newcommand{\gI}{\mathfrak{I}}
\newcommand{\gV}{\mathcal{V}}

\newcommand{\bu}{\bullet}

\newcommand {\be}{\begin{equation}}
\newcommand {\ee}{\end{equation}}
\newcommand{\h}{\begin{eqnarray*}}
\newcommand{\e}{\end{eqnarray*}}

\begin{document}


\title[Exotic twisted equivariant cohomology of loop spaces]
{Exotic twisted equivariant cohomology of loop spaces, \\ twisted Bismut-Chern character and T-duality}


\author{Fei Han}
\address{Department of Mathematics,
National University of Singapore, Singapore 119076}
\email{mathanf@nus.edu.sg}

 \author{Varghese Mathai}
\address{School of Mathematical Sciences,
University of Adelaide, Adelaide 5005, Australia}
\email{mathai.varghese@adelaide.edu.au}

\date{}
\thanks{ The first author was partially supported
by the grant AcRF R-146-000-163-112 from National University
of Singapore. He is grateful to Professor Weiping Zhang for inspiring discussions. 
The second author was supported by funding from the Australian Research Council, through Discovery Projects
DP110100072 and DP130103924. The second author also thanks the Department of Mathematics,
National University of Singapore, for hospitality during a recent visit when this research was initiated.
}

\maketitle

\begin{abstract}
We define 
{\em exotic twisted $\TT$-equivariant cohomology} for the loop space $LZ$ of a smooth manifold $Z$
via the invariant differential forms on $LZ$ with coefficients in the (typically non-flat) holonomy line bundle of a gerbe, with differential an equivariantly flat superconnection. We introduce the twisted Bismut-Chern character form, a loop space refinement of the twisted Chern character form in \cite{BCMMS},
which represent classes in the completed periodic {exotic twisted $\TT$-equivariant cohomology} of $LZ$. We establish a localisation theorem for the completed periodic exotic twisted $\TT$-equivariant cohomology for loop spaces and apply it to establish T-duality in a background flux  in type II String Theory from a loop space perspective.
\end{abstract}

\tableofcontents

\section*{Introduction}

This paper is motivated in part by Witten \cite{W82}, Atiyah \cite{A85} and Bismut \cite{B85}.  Atiyah, working out an idea of Witten, revealed the remarkable fact that the index of the Dirac operator on the spin complex of a spin manifold can be formally interpreted as an integral of an equivariantly closed (with respect to the standard circle action on the loop space) differential form over loop space. A formal application of the localisation formula of Duistermaat-Heckman \cite{DH82, DH83} leads to the index theorem of Atiyah-Singer for the Dirac operator. Bismut extended this approach to a Dirac operator twisted by a vector bundle with connection. In doing so, for a vector bundle with connection, he constructed an equivariantly closed form on the loop space, lifting the Chern character form of the vector bundle with connection to the loop space. 
This paper is also partly motivated by recent developments in String Theory in a background flux. In \cite{BM00}, it was argued that D-brane charges in a background $H$-flux take values in twisted
K-theory of spacetime $Z$, $K^\bullet(Z,H)$. The Chern-Weil representatives of the twisted Chern character
$Ch_H:  K^\bullet(Z,H) \to H^\bullet(Z,H)$ were defined and its properties studied in \cite{BCMMS, MS}.
In this paper, we use some of the fundamental constructions in \cite{B85} and extend them to twisted K-theory and twisted cohomology \cite{BCMMS, MS}.

In \S1 and \S 2, we introduce a 
completed periodic exotic twisted $\TT$-equivariant cohomology of loop space
$h_\TT^\bullet(LZ, \nabla^{\cL^B}:\bar H)$
via the  invariant differential forms on $LZ$ with coefficients in the (typically non-flat) holonomy line bundle $\cL^B$, denoted
$\Omega^\bullet(LZ, \cL^B)^\TT,$
with canonical connection $\nabla^{\cL^B}$ of a gerbe $\cG_B$ with connection on $Z$, with differential the equivariantly flat superconnection
$\nabla^{\cL^B} - i_K + \bar H$, where $\bar H$ is an explicit $LZ$ extension of $H$. We then introduce the twisted Bismut-Chern character form, $BCh_H(E, \nabla^E)$, which represents classes in
$h_\TT^\bullet(LZ, \nabla^{\cL^B}:\bar H)$ and is a loop space refinement of the twisted Chern character form $Ch_H(E, \nabla^E)$ in \cite{BCMMS, MS}. When $H=0$, it reduces to the Bismut-Chern character form $BCh(E, \nabla^E)$ in \cite{B85}: see also 
\cite{FH08,TWZ,S,LMRT, GJP} for other interesting interpretations and extensions of the Bismut-Chern character. More precisely, we define these in such a way that the following diagram commutes,
\beq\label{twistedBC_H}
\xymatrix{K^\bullet(Z,H)\ar[dr]_{Ch_H}\ar[rr]^{BCh_H}&&
h_\TT^\bullet(LZ, \nabla^{\cL^B}:\bar H)\ar[dl]^{res}\\
&H^\bu(\Omega(Z)[[u, u^{-1}]], d+u^{-1}H)&}
\eeq
where $res$ is the localisation map, $\mathrm{degree}(u)=2$ and $\bar H$ is a degree 3-form on $LZ$ that is completely determined by $H$
which will be defined in the next section. We establish the basic properties of completed periodic exotic twisted $\TT$-equivariant cohomology in \S2 and define the twisted Bismut-Chern character $BCh_H$ in \S3. This includes the Localisation Theorem in \S2, as well as Lemma \ref{lem:equiv}, establishing the isomorphisms
\beq
h_\TT^\bullet(LZ, \nabla^{\cL^B}:\bar H) \cong H^\bu(\Omega(Z)[[u, u^{-1}]], d+u^{-1}H) \cong H^\bu(Z, H)[[u, u^{-1}]].
\eeq
This justifies our premise {\em that over the rationals at least, D-brane charges in a background H-flux take values in $h_\TT^\bullet(LZ, \nabla^{\cL^B}:\bar H),$}
whose configuration space is the space of loops on spacetime $Z$. This perspective is more natural in the context of String Theory.

We apply this in \S4 to the study of T-duality.
In  \cite{BEM04a, BEM04b}, the following situation is studied.
\begin{equation*}
\xymatrix{(Z,H)\ar[dr]_{p}&&
(\widehat Z, \widehat H)\ar[dl]^{\widehat p}\\
& X&}
\end{equation*}
where $Z, \widehat Z$ are principal circle bundles over a base $X$ with fluxes $H$ and $\widehat H$, respectively, satisfying
$p_*(H)=c_1(\widehat Z), \, \widehat p_*(\widehat H)=c_1(Z)$ and $H-\widehat H$ is exact on the
correspondence space $Z\times_X \widehat Z$. Then in \cite{BEM04a, BEM04b}, it was proved that
there is an isomorphism of twisted K-theories $$K^\bullet(Z,H) \cong K^{\bullet+1}(\widehat Z, \widehat H)$$
and an isomorphism of twisted cohomology theories, $H^\bullet(Z,H) \cong H^{\bullet+1}(\widehat Z, \widehat H),$
so by Lemma \ref{lem:equiv}, also an equivalence
$$H^\bu(\Omega(Z)[[u, u^{-1}]], d+u^{-1}H) \cong 
H^{\bu+1}(\Omega(\widehat{Z})[[u, u^{-1}]], d+u^{-1}\widehat{H}).$$
As a consequence of our Localisation Theorem \ref{thm:localisation}, properties of the twisted Bismut-Chern 
character in \S3, and the T-duality isomorphisms described above in \cite{BEM04a, BEM04b}, 
we obtain a T-duality isomorphism of completed periodic exotic twisted $\TT$-equivariant cohomologies,
$h_\TT^\bullet(LZ, \nabla^{\cL^B}:\bar H) \cong h_\TT^{\bullet+1}(L{\widehat Z}, \nabla^{\cL^{\widehat B}}:\bar{\widehat H}),$ which is  interpreted as {\em T-duality from a loop space perspective}, giving an equivalence (rationally) between 
D-brane charges in a background H-flux in type IIA and IIB string theories. As a consequence of the above, 
we propose that the configuration space 
of Ramond-Ramond fields be the space of differential forms with coefficients in the holonomy line bundle 
 on loop space $LZ$ of the gerbe $\cG_B$ on spacetime $Z$, and that are closed under the equivariantly closed superconnection.


\section{Exotic twisted $\TT$-equivariant cohomology of loop space}

Let $Z$ be a smooth manifold. In this section, we define the exotic twisted $\TT$-equivariant cohomology of the loop space $LZ$.

\subsection{Brylinski open cover} We want to study open covers $\{U_\alpha\}$ of $Z$ such that the space of loops
$\{LU_\alpha\}$ is an open cover of $LZ$. The usual Cech open cover of $Z$ consisting of a convex open cover
of $Z$ does {\em not} satisfy this property. Suppose that $\{U_\alpha\}$ is a maximal open cover of $Z$ with the property that
$H^i(U_{\alpha_I})=0$ for $i=2, 3$ where $U_{\alpha_I} = \bigcap_{i\in I} U_{\alpha_i}, \,$ $|I|<\infty.$
We call such an open cover a {\em Brylinski open cover} of $Z$.
It is easy to see that $\{LU_\alpha\}$ is an open cover of $LZ$. For instance let $\gamma:\TT\to Z$ be a smooth loop in $Z$ and $U_\gamma$ a tubular neighbourhood of $\gamma$ in $Z$. Then $\{U_\gamma, \gamma\in LZ\}$ is a Brylinski open cover.

\subsection{Cech-de Rham cocycles on Brylinski open covers and gerbes} Let $\{U_\alpha\}$ be a Brylinski open cover of $Z$ and
$H$ a closed 3-form on $Z$ such that $\frac{1}{2\pi i}H$ has integral periods. Since $H^3(U_\alpha)=0$,  $H\big|_{U_\alpha} = dB_\alpha$  for some purely imaginary
$B_\alpha \in \Omega^2(U_\alpha)$. Also $B_\beta-B_\alpha = dA_{\alpha\beta}$ since $H^2(U_\alpha \cap U_\beta)=0$, where $A_{\alpha\beta}$ is also chosen to be purely imaginary.
Then $(H, B, A)$ defines a connective structure (or connection) for a gerbe $\cG_B$ on $Z$.

A geometric realization of the gerbe $\cG_B$ is $\{(L_{\alpha\beta}, \nabla^L_{\alpha\beta})\}$, a collection of line bundles $L_{\alpha\beta}\to U_{\alpha\beta}$ such that there is an isomorphism
$L_{\alpha\beta} \otimes L_{\beta\gamma} \cong L_{\alpha\gamma}$ on $U_{\alpha\beta\gamma}$ and collection of connections
$\{\nabla^L_{\alpha\beta}\}$ such that $\nabla^L_{\alpha\beta}=d+A_{\alpha\beta}$  (note that as $H^2(U_\alpha \cap U_\beta)=0$, the bundle $L_{\alpha\beta}$ is trivial). Then we have
\beq\label{gerbeconn}
(\nabla^L_{\alpha\beta})^2 = F^L_{\alpha\beta} = B_\beta- B_\alpha.
\eeq

The holonomy of this gerbe is a
line bundle $\cL^B \to LZ$ over the loop space $LZ$. $\cL^B$ has Brylinski local sections $\{\sigma_\alpha\}$ with respect to $\{LU_\alpha\}$ such that the transition functions are $\{e^{-\int_0^1i_K A_{\alpha_\beta}} = e^{-\tau(A_{\alpha\beta})}\}$, where $\tau$ is the transgression defined in \eqref{eqn:trans}, i.e. $\sigma_\alpha=e^{-\int_0^1i_K A_{\alpha_\beta}}\sigma_\beta$, where $K$ is the vector field on $LZ$ generating rotation of loops. The Brylinski sections are $\TT$-invariant.  $\cL^B$ comes with a natural connection, whose definition with respect to the basis $\{\sigma_\alpha\}$ is
\beq
\nabla^{\cL^B} = d-i_K \bar B_\alpha=d-\tau(B_\alpha),
\eeq
where $ \bar B_\alpha$ is explained below and
\beq\label{eqn:trans}
\tau: \Omega^\bullet(U_{\alpha_I} ) \longrightarrow \Omega^{\bullet-1}(LU_{\alpha_I} )
\eeq
is the transgression map defined as
\beq
\tau(\xi_I) = \int_{\TT} ev^*(\xi_I), \qquad \xi_I \in \Omega^\bullet(U_{\alpha_I} ).
\eeq
Here $ev$ is the evaluation map
\beq
ev: \TT \times LZ \to Z: (t, \gamma)\to \gamma(t).
\eeq
The curvature of the connection $\nabla^{\cL^B}$ is $F_B = (\nabla^{\cL^B})^2 =- \tau(H)$ is the transgression of the minus 3-curvature
$H$ of the gerbe $\cG_B$.

For more details, cf. \cite{Bry}.

\subsection{Exotic twisted $\TT$-equivariant cohomology of loop space}
Consider $\Omega^\bullet(LZ, \cL^B)$ which is the space of differential forms on loop space $LZ$ with values in
the holonomy line bundle $\cL^B \to LZ$ of the gerbe $\cG_B$ on $Z$.

Let $\omega \in \Omega^i(Z)$. Define $\hat\omega_s \in \Omega^i(LZ)$ for $s\in [0,1]$ by
\beq
\hat\omega_s(X_1, \ldots, X_i)(\gamma) = \omega(X_1\big|_{\gamma(s)}, \ldots, X_i\big|_{\gamma(s)})
\eeq
for $\gamma\in LZ$ and $X_1, \ldots, X_i$ are vector fields on $LZ$ defined near $\gamma$. Then one checks that
$d \hat\omega_s = {\widehat{d\omega}}_s$. The $i$-form
\beq
\bar\omega=\int_{0}^1  \hat\omega_s ds \in \Omega^i(LZ)
\eeq
is $\TT$-invariant, that is, $L_K\left(\bar\omega\right) = 0$. Moreover $\tau(\omega) = i_K\bar\omega$.

Let $H$ be a closed 3-form on $Z$ with integral periods as above, and $\bar H \in \Omega^3(LZ)$
be the associated closed 3-form on $LZ$.
Define $D_{\bar H} = \nabla^{\cL^B} - i_K + \bar H$. Then we compute,

\begin{lemma} $(D_{\bar H})^2=0$ on $\Omega^\bullet(LZ, \cL^B)^\TT$.
\end{lemma}
\begin{proof} Let $\{U_\alpha\}$ be a Brylinski open cover of $Z$.
Then $\bar H\Big|_{LU_\alpha} = d\bar B_\alpha$ on $LU_\alpha$.
On $LU_\alpha$, we have
\begin{align}
&(D_{\bar H})^2 \\
=& ( \nabla^{\cL^B} - i_K + \bar H)^2\\
=& (d -i_K\bar B_\alpha - i_K + \bar H)^2\\
=& \left( (d-i_K) + (d-i_K)\bar B_\alpha\right)^2\\
=& \left( \exp(-\bar B_\alpha)(d-i_K) \exp(\bar B_\alpha) \right)^2\\
=& - L_K - (L_K\bar B_\alpha) = -L_K,
\end{align}
where $L_K$ denotes the Lie derivative of the vector field $K$.
 As the Brylinski sections are invariant, we have $L_K=L^{\cL^B}_K$ on $LU_\alpha$. So $(D_{\bar H})^2 = -L^{\cL^B}_K$,
which vanishes on
$\Omega^\bullet(LZ, \cL^B)^\TT$ as claimed.
\end{proof}

Notice that $D_{\bar H} =\nabla^{\cL^B} - i_K + \bar H$ is a {\em flat $\TT$-equivariant
superconnection} (in the sense of Quillen \cite{Q, MQ}) on $\Omega^\bullet(LZ, \cL^B)^\TT$. Therefore
$\left(\Omega^\bullet(LZ, \cL^B)^\TT, D_{\bar H}\right)$ is a $\ZZ_2$-graded complex.  We call the cohomology
of this complex the {\em exotic twisted $\TT$-equivariant cohomology} of loop space, denoted by $H_\TT^\bullet(LZ, \nabla^{\cL^B}:\bar H)$.

\smallskip

\section{Localisation theorem}
In this section, we prove the localisation theorem in our context. Much of the background material can be found in \cite{JP}.

Let $M$ be a (possibly infinite dimensional) $\TT$-manifold and $\xi$ a $\TT$-equivariant complex line bundle over $M$. Let $\nabla^\xi$  be a $\TT$-invariant connection on $\xi$  and $H\in \Omega^3_{cl}(M)$ a closed 3-form such that $$(\nabla^{\xi}-ui_K+u^{-1}H)^2+uL_K^\xi=0,$$ where $K$ is the Killing vector field.

Define the {\it completed periodic exotic twisted $\TT$-equivariant cohomology} $h^\bu_\TT(M, \nabla^\xi:H)$ to be the cohomology of the complex
\be (\Omega^\bu(M, \xi)^\TT[[u, u^{-1}]], \nabla^\xi-ui_K+u^{-1}H). \ee

A $\TT$-manifold $M$ is called {\em strongly regular} if the fixed point set $F$ is a smooth submanifold, which has an invariant neighbourhood $N$ such that (i) $N$ has an invariant good cover $\{W_\alpha\}$, i.e. each $W_\alpha$ is $\TT$-homotopic to a point; (ii) there exists a projection $p:N\to F$ and an equivariant homotopy
$$g: N\times I\to N, $$
($I=[0,1]$) with the property that
$$g_0=i\circ p,\quad g_1=id, $$
$$p(g_t(x))=g_0(x), \quad\forall x\in N, t\in I,$$ where $i:F\to N$ is the embedding.

\begin{theorem}[Localisation Theorem]\label{thm:localisation} If $M$ is a strong regular $\TT$-manifold, then
 \be i^*: h^\bu_\TT(M, \nabla^\xi:H)\cong h^\bu_\TT(F, i^*\nabla^\xi:i^*H)=H(\Omega^\bu(F, i^*\xi)[[u, u^{-1}]], i^*\nabla^\xi+u^{-1}i^*H).\ee
 is an isomorphism, where for simplicity, we also denote the embedding of $F$ into $M$ by $i$.
\end{theorem}

We apply this to the following situation. Let $Z$ be a finite dimensional manifold, and let $LZ$ denote the infinite dimensional manifold consisting of all smooth loops in $Z$.
Let $\TT$ act on $LZ$ by rotation of loops. $Z$ has an invariant tubular neighbourhood in $LZ$, this makes $LZ$ into a strongly regular $\TT$-manifold (cf. \cite{JP})
with fixed point set $(LZ)^\TT=Z$, identified with the space
of all the constant loops in $Z$.
As the holonomy line bundle $\cL^B$ is trivial when restricted to $Z$, the constant loop space,  so by the localisation theorem above, we have

\begin{corollary} The restriction to the constant loops
\be res: h^\bu_\TT(LZ, \nabla^{\cL^B}: \bar H)\cong H(\Omega^\bu(Z)[[u, u^{-1}]], d+u^{-1}H).\ee
is an isomorphism.
\end{corollary}

This justifies the following proposal:

\begin{quote}
{\em Over the rationals, D-brane charges on space-time $Z$ in a background $H$-flux, take values in $h^\bu_\TT(LZ, \nabla^B: \bar H)$.}
\end{quote}
Note that the configuration space in this proposal is the much larger space of differential forms on loop space $LZ$, which is more naturally related to the geometric picture of strings. \smallskip

Before we prove Theorem 2.1, let us first establish some properties of completed periodic exotic twisted $\TT$-equivariant cohomology. 

\begin{lemma} Let $M$ be a $\TT$-manifold and $\xi$ be an $\TT$-equivariant complex line bundle over $M$. Let  $\nabla^\xi $ be a connection on $\xi$  and $H\in \Omega^3(M)$ be a $3$-form such that $$(\nabla^{\xi}-ui_K+u^{-1}H)^2+uL_K^\xi=0.$$ Suppose the action has no fixed points, then
$$h_\TT^\bu(M, \nabla^\xi: H)=0. $$
\end{lemma}
\begin{proof} Choose a $\TT$-invariant metric on $TM$.  Let $\theta$ be the one form dual to the Killing vector field $K$. Then we have
\be
\begin{split} &(d-ui_K)\theta\\
=&d\theta-ui_K\theta\\
=&d\theta-u|K|^2\\
=&-u|K|^2\left(1-\frac{d\theta}{u|K|^2}.\right)
\end{split}
\ee

Let $\gamma=((d-ui_K)\theta)^{-1}=-u^{-1}|K|^{-2}
\sum_{i=0}^\infty\left(\frac{d\theta}{u|K|^2}\right)^i.$ As $\gamma((d-ui_K)\theta)=1$, applying $d-ui_K$ on both sides and using the fact that $L_K\theta=0$, we have
$(d-ui_K)\gamma=0$.

Define $\omega=\theta\gamma$, an odd degree form.  Then
$(d-ui_K)\omega=((d-ui_K)\theta)\gamma=1.$

 $\forall x\in \Omega^\bu(M, \xi),$ we have
 \be
 \begin{split}
 &(\nabla^{\xi}-ui_K+u^{-1}H)(\omega\cdot x)+\omega\cdot( (\nabla^{\xi}-ui_K+u^{-1}H)x)\\
 =& (\nabla^{\xi}-ui_K)(\omega\cdot x)+(u^{-1}H\omega)\cdot x+\omega\cdot( (\nabla^{\xi}-ui_K)x)+(u^{-1}\omega H)\cdot x\\
 =&(\nabla^{\xi}-ui_K)(\omega\cdot x)+\omega\cdot( (\nabla^{\xi}-ui_K)x)\\
 =&((d-ui_K)\omega)\cdot x-\omega\cdot ((\nabla^{\xi}-ui_K)x)+\omega\cdot( (\nabla^{\xi}-ui_K)x)\\
 =&x.
 \end{split}
 \ee
This homotopy tells us that  $h_\TT^\bu(M, \nabla^\xi: H)=0. $
\end{proof}

Let $M$ be a $\TT$-manifold and $\xi$ be an $\TT$-equivariant complex line bundle over $M$. Let  $\nabla^\xi $ be a connection on $\xi$  and $H\in \Omega^3(M)$ be a $3$-form such that $$(\nabla^{\xi}-ui_K+u^{-1}H)^2+uL_K^\xi=0.$$  Suppose $M=U\cup V$, with $U, V$ open. Assume $U, V$ and $U\cap V$ are all $\TT$-invariant submanifolds.  Thanks to the partition of unity, the following sequence is exact,
\be \label{shortMV} 0\to \Omega^\bu(M, \xi)^\TT[[u, u^{-1}]]\to  \Omega^\bu(U, \xi)^\TT[[u, u^{-1}]]\oplus\Omega^\bu(V,\xi)^\TT[[u, u^{-1}]]\to \Omega^\bu(U\cap V, \xi)^\TT[[u, u^{-1}]]\to 0. \ee
$$\ \ \ \ \ \ \ \ \ \ \ \ \ \ \\ \ \ \ \ \ \ \ \ \ \ \ \ \ \ \ \ \ \ \ \\ \ \ \ \ \ (\omega, \tau)\ \ \ \ \ \ \ \ \ \ \ \ \ \ \\ \ \ \ \ \ \ \\\mapsto (\tau-\omega) $$
Actually, take any $f\in \Omega^\bu(U\cap V, \xi)^\TT$, let $\{\rho_U, \rho_V\}$ be a partition of unity subordinate to the open cover $\{U, V\}$. We can assume that $\rho_U, \rho_V$ are both $\TT$-invariant, otherwise, simply average them over $\TT$.  Then $(-\rho_Vf,\rho_Uf)\in  \Omega^\bu(U, \xi)^\TT\oplus\Omega^\bu(V,\xi)^\TT$ maps onto $f$.

\begin{lemma} The short exact sequence (\ref{shortMV}) induces a long exact Mayer-Vietoris sequence in the completed periodic exotic twisted $\TT$-equivariant cohomology,

\be
 \xymatrix{
h_\TT^{ev}(M,\nabla^\xi:H) \ar[r]& h_\TT^{ev}(U,\nabla^\xi:H) \oplus h_\TT^{ev}(V,\nabla^\xi:H) \ar[r]& h_\TT^{ev}(U\cap V,\nabla^\xi:H) \ar[d]^{D^*}\\
h_\TT^{odd}(U\cap V,\nabla^\xi:H)\ar[u]^{D^*} &  \ar[l] h_\TT^{odd}(U,\nabla^\xi:H) \oplus h_\TT^{odd}(V,\nabla^\xi:H) & \ar[l] h_\TT^{odd}(M,\nabla^\xi:H)
}
\ee
\end{lemma}

\begin{proof} Let $\omega\in \Omega^{ev}(U\cap V, \xi)^\TT[[u, u^{-1}]]$ such that $(\nabla^{\xi}-ui_K+u^{-1}H)\omega=0$. Let $\{\rho_U, \rho_V\}$ be a partition of unity subordinate to the open cover $\{U, V\}$ such that $\rho_U, \rho_V$ are both $\TT$-invariant.

Define the coboundary operator by
\be
D^*([\omega])=\left\{
\begin{array}{ccc}
[-(\nabla^{\xi}-ui_K+u^{-1}H)(\rho_V\omega)]& \mathrm{on} \ U\\

[(\nabla^{\xi}-ui_K+u^{-1}H)(\rho_U\omega)]& \mathrm{on} \ V
\end{array}
\right.
\ee
which is easily seen to be an element in $h_\TT^{odd}(M,\nabla^\xi:H)$ and is independent of choices in this construction.  Similarly one can define the coboundary operator on $h_\TT^{odd}(U\cap V,\nabla^\xi:H)$.

It is not hard to check the exactness of the sequence.
\end{proof}

\begin{lemma}Let $M$ be a $\TT$-manifold and $\xi$ be an $\TT$-equivariant complex line bundle over $M$. Let  $\nabla^\xi $ be a connection on $\xi$  and $H\in \Omega^3(M)$ be a $3$-form such that $$(\nabla^{\xi}-ui_K+u^{-1}H)^2+uL_K^\xi=0.$$
Let $i_0:M \rightarrow M\times I,  m \mapsto (m,0)$ be the inclusion and $\pi: M\times I\to M$ be the projection. $\TT$ acts on $M\times I$ in the obvious way. We have
$$h^\bu_\TT(M\times I, \pi^*\nabla^{\xi}:\pi^*H)\cong h^\bu_\TT(M, \nabla^\xi:H). $$

\end{lemma}

\begin{proof}It is clear that $i_0^*\circ \pi^*=id: \Omega^\bu(M, \xi)\to \Omega^\bu(M, \xi).$

We will show that $\pi^*\circ i_0^*$ is homotopic to identity on $\Omega^\bu(M\times I, \pi^*\xi)$.

Choose an atlas $\{U_\alpha\}$ for $M$, then $\{U_\alpha\times I\}$ is an atlas on $M\times I$. Let $\{s_\alpha\}$ be local basis of $\xi$ on $\{U_\alpha\}$ and $\{g_{\alpha\beta}\}$ be the transition functions. Then $\{\pi^*s_\alpha\}$ are local basis for the bundle $\pi^*\xi$ on $\{U_\alpha\times I\}$ and $\{\pi^*g_{\alpha\beta}\}$ are the transition functions.

For any $\omega\in \Omega^\bu(M\times I, \pi^*\xi)$, define $K_0\omega \in \Omega^\bu(M\times I, \pi^*\xi)$ in the following way. Let $\omega=\omega_\alpha\otimes (\pi^*s_\alpha)$ on $U_\alpha\times I$. If $\omega_\alpha$ is of the form $(\pi^*\psi)f(x,t)$, set $K_0(\omega)|_{U_\alpha\times I}=0$; if $\omega_\alpha$ is of the form $(\pi^*\psi)f(x,t)dt$, set $K_0(\omega)|_{U_\alpha\times I}=((\pi^*\psi)\int_0^tf(x,t)dt)\otimes (\pi^*s_\alpha).$ As the transition function from $\pi^*s_\alpha$ to $\pi^*s_\beta$ is $\pi^*g_{\alpha\beta}$, it is not hard to see that $K_0(\omega)|_{U_\alpha\times I}$ patch together to give $K_0(\omega)\in \Omega^\bu(M\times I, \pi^*\xi)$.

By (\cite{BT}, Sec 4), we know that on $U_\alpha\times I$, (for simplicity, we also denote by $K_0$ the similar operator on $\Omega^\bu(M\times I)$),
\be (1-\pi^*\circ i_0^*)\omega_\alpha=(-1)^{p(\omega_\alpha)-1}(dK_0-K_0d)\omega_\alpha.  \ee
Since the connection $\pi^*\nabla^{\xi}$ on $\pi^*\xi$ is horizontal, we have
\be (1-\pi^*\circ i_0^*)\omega=(-1)^{p(\omega)-1}((\pi^*\nabla^{\xi})\circ K_0-K_0\circ(\pi^*\nabla^{\xi}))\omega.  \ee

Moreover, as $i_K$ and $\pi^*H$ are both horizontal, they both commute with $K_0$ and therefore,
\be (1-\pi^*\circ i_0^*)\omega=(-1)^{p(\omega)-1}((\pi^*\nabla^{\xi}-ui_K+u^{-1}\pi^*H)\circ K_0-K_0\circ(\pi^*\nabla^{\xi}-ui_K+u^{-1}\pi^*H))\omega.  \ee

The isomorphism therefore follows.

\end{proof}

\begin{lemma}Let $M$ be a $\TT$-manifold with an invariant good cover $\{W_\alpha\}$. $\TT$ acts on $M\times I$ in the obvious way.
Let $\xi$ be an $\TT$-equivariant complex line bundle over $M\times I$ equipped with a $\TT$-invariant connection $\nabla^\xi $ and $H\in \Omega^3_{cl}(M\times I)$ a closed $3$-form on $M\times I$ such that $$(\nabla^{\xi}-ui_K+u^{-1}H)^2+uL_K^\xi=0.$$
Let $i_0:M \rightarrow M\times I,  m \mapsto (m,0)$ and $i_1:M \rightarrow M\times I,  m \mapsto (m,1)$ be the inclusions. We have
$$h^\bu_\TT(M, i_0^*\nabla^{\xi}:i_0^*H)\cong h^\bu_\TT(M, i_1^*\nabla^{\xi}:i_1^*H). $$

\end{lemma}
\begin{proof} Let $\{s_\alpha\}$ be local invariant basis of $i_0^*\xi$ on $\{W_\alpha\}$ and $\{g_{\alpha\beta}\}$ be the transition functions. Let $\Theta_\alpha$ be the basis of bundle $\xi$ on $W_\alpha\times I$ such that
\be
\left\{
\begin{array}{cc}
&\nabla^\xi_{\frac{\partial}{\partial t}}\Theta_\alpha=0\\
&\Theta_\alpha|_{M\times \{0\}}=s_\alpha.
\end{array}
\right.
\ee
Let $\pi: M\times I\to M$ be the projection.  Since $\pi^*g_{\alpha\beta}$ are horizontal, we can see that $\{\pi^*g_{\alpha\beta}\}$ are transition functions of the local basis $\{\Theta_\alpha\}$. Let $\{\theta_\alpha\}$ be the connection one form of $\nabla^{\xi}$ with respect to the local basis $\{\Theta_\alpha\}.$ Obviously, they are horizontal forms. Also as $s_\alpha$ are invariant sections and $\nabla^\xi$ is an invariant connection, $\Theta_\alpha$ are also local invariant basis of $\xi$ with respect to the cover $\{W_\alpha\times I\}$.

Take any $\omega\in \Omega^\bu(M, i_0^*\xi)$. Let $\omega=\omega_\alpha\otimes s_\alpha$ on $W_\alpha$. Define
$$\tau(\omega)|_{U_\alpha\times I}=(\pi^*\omega_\alpha)\otimes \Theta_\alpha.$$ It is clear that $\tau(\omega)|_{W_\alpha\times I}$ patch together to give $\tau(\omega)\in \Omega^\bu(M\times I, \xi).$

Define
\be \rho^0_1: \Omega^\bu(M, i_0^*\xi)\to  \Omega^\bu(M, i_1^*\xi) \ee
$$\omega\mapsto e^{-u^{-1}K_0H}\tau(\omega)|_{M\times \{1\}}, $$
where $K_0$ is defined as in the proof of Lemma 2.5.

Suppose $(i_0^*\nabla^\xi-ui_K)\omega=-u^{-1}(i_0^*H)\cdot \omega.$ Locally this means that
\be (d+i_0^*\theta_\alpha-ui_K)\omega_\alpha=-u^{-1}(i_0^*H)\cdot \omega_\alpha.  \ee
Then if we compute $(\nabla^\xi-ui_K)(e^{-u^{-1}K_0H}\tau(\omega)) $, locally under the basis $\Theta_\alpha$, we have
\be
\begin{split}
&(d+\theta_\alpha-ui_K)(e^{-u^{-1}K_0H}\pi^*(\omega_\alpha))\\
=&e^{-u^{-1}K_0H}(-u^{-1}d(K_0H))\pi^*\omega_\alpha+e^{-u^{-1}K_0H}\pi^*(d\omega_\alpha)\\
&+(\theta_\alpha-\pi^*i_0^*\theta_\alpha)e^{-u^{-1}K_0H}\pi^*(\omega_\alpha)+e^{-u^{-1}K_0H}\pi^*(i_0^*\theta_\alpha\omega_\alpha)\\
&+e^{-u^{-1}K_0H}(i_KK_0H)\pi^*(\omega_\alpha)-ue^{-u^{-1}K_0H}\pi^*(i_K\omega_\alpha)\\
=&e^{-u^{-1}K_0H}[(-u^{-1}d(K_0H)-u^{-1}\pi^*\circ i_0^*H)+(\theta_\alpha-\pi^*\circ i_0^*\theta_\alpha+i_KK_0H)]\pi^*\omega_\alpha.
\end{split}
\ee

However by homotopy formula, we have $ H-\pi^*\circ i_0^*H=dK_0H-K_0dH=dK_0H.$ So
$dK_0H+\pi^*\circ i_0^*H=H$.

Moreover, by homotopy formula, $\theta_\alpha-\pi^*\circ i_0^*\theta_\alpha=dK_0\theta_\alpha-K_0d\theta_\alpha. $ As $\theta_\alpha$ is horizontal, $K_0\theta_\alpha=0$. So $\theta_\alpha-\pi^*\circ i_0^*\theta_\alpha=-K_0d\theta_\alpha.$

But
$(\nabla^{\xi}-ui_K+u^{-1}H)^2+uL_K^\xi=0$ tells us that $d\theta_\alpha-i_KH=0$. Noticing that $K$ is horizontal,  so $i_KK_0=K_0i_K$. Therefore we have
\be \theta_\alpha-\pi^*\circ i_0^*\theta_\alpha+i_KK_0H=-K_0d\theta_\alpha+K_0i_KH=-K_0d\theta_\alpha+K_0d\theta_\alpha=0. \ee
So
\be (d+\theta_\alpha-ui_K)(e^{-u^{-1}K_0H}\pi^*(\omega_\alpha))=-u^{-1}H(e^{-u^{-1}K_0H}\pi^*(\omega_\alpha)),\ee
which gives
\be  (\nabla^\xi-ui_K)(e^{-u^{-1}K_0H}\tau(\omega))=-u^{-1}H(e^{-u^{-1}K_0H}\tau(\omega)). \ee

This shows that $\rho_1^0$ gives us a homomorphism
\be  \rho_1^0: h^\bu_\TT(M, i_0^*\nabla^{\xi}:i_0^*H)\to h^\bu_\TT(M, i_1^*\nabla^{\xi}:i_1^*H).  \ee

It is easy to see that we can define the inverse $\rho_0^1$ in a similar manner. Just need to replace $K_0$ by $K_1$ with replacing $(\pi^*\psi)\int_0^tf(x,t)dt$ by $(\pi^*\psi)\int_t^0f(x,t)dt$ in the definition of $K_0.$
\end{proof}

\begin{remark}We would like to point out that there exist $\TT$-manifolds with invariant good covers. For instance, the invariant tubular neighborhoods of a finite dimensional manifold in its loop space.  
\end{remark}

In the following, we will give a proof of Theorem 2.1.

Pick an invariant tubular neighbourhood of $N$ of $F$ such that there exists a projection $p:N\to F$ and an equivariant homotopy
$$g: N\times I\to N, $$ with the property that
$$g_0=i\circ p, g_1=id, $$
$$p(g_t(x))=g_0(x), \forall x\in N, t\in I. $$

As $i\circ p=id$, we have
$$i^*p^*=id: h^\bu_\TT(F, i^*\nabla^\xi, i^*H) \to h^\bu_\TT(N, p^*i^*\nabla^\xi,p^* i^*H)\to h^\bu_\TT(F, i^*\nabla^\xi, i^*H).$$

As $p(g_t(x))=g_0(x), \forall x\in N, t\in I, $ we have $i\circ p\circ g=i\circ p\circ \pi. $ So
\be g^*p^*i^*(\xi)=\pi^*(p^*i^*(\xi)),   g^*p^*i^*(H)= \pi^*(p^*i^*(H)).          \ee

By Lemma 2.5, we have
\be h^\bu_\TT(N, p^*i^*\nabla^\xi: p^*i^*H)\cong h^\bu_\TT(N\times I, g^*p^*i^*(\nabla^\xi): g^*p^*i^*(H)).\ee

Suppose $i_0:N \rightarrow N\times I,  m \mapsto (m,0)$ and $i_1:N \rightarrow N\times I,  m \mapsto (m,1)$ be the inclusions. Since $i\circ p=g\circ i_0$, we have
\be i_0^*g^*=p^*i^*: h^\bu_\TT(N, p^*i^*\nabla^\xi: p^*i^*H)\to h^\bu_\TT(F, i^*\nabla^\xi, i^*H) \to h^\bu_\TT(N, p^*i^*\nabla^\xi: p^*i^*H). \ee
By Lemma 2.5, $i_0^*$ and $i_1^*$ are both inverse to
$$\pi^*:h^\bu_\TT(N, p^*i^*\nabla^\xi: p^*i^*H)\to h^\bu_\TT(N\times I, g^*p^*i^*(\nabla^\xi): g^*p^*i^*(H)). $$ So
$$p^*i^*=i_1^*g^*=(g\circ i_1)^*=id:  h^\bu_\TT(N, p^*i^*\nabla^\xi: p^*i^*H)\to h^\bu_\TT(N, p^*i^*\nabla^\xi: p^*i^*H).$$

Therefore we have
\be  h^\bu_\TT(F, i^*\nabla^\xi, i^*H) \cong h^\bu_\TT(N, p^*i^*\nabla^\xi,p^* i^*H).   \ee

On the other hand, $g^*\xi, g^*(\nabla^\xi), g^*(H)$ have the property that
$$g^*H\in \Omega^3_{cl}(N\times I), \ (g^*\nabla^\xi-ui_K+u^{-1}g^*H)^2+L_K^{g^*\xi}=0.$$
Also $$i_0^*g^*\xi=p^*i^*\xi, i_0^*(g^*(H))=p^*i^*(H),$$
$$i_1^*g^*\xi=\xi, i_1^*(g^*(H))=H.$$
Then by Lemma 2.6, we have
\be  h^\bu_\TT(N, p^*i^*\nabla^\xi,p^* i^*H)\cong h^\bu_\TT(N, \nabla^\xi, H) .   \ee

So we can see that
\be h^\bu_\TT(F, i^*\nabla^\xi, i^*H) \cong h^\bu_\TT(N, \nabla^\xi, H).\ee

Note that $Y=(Y\setminus F)\cup N$ and $Y\setminus F$, $(Y\setminus F)\cap N=N\setminus F$ are fixed point free. Then Lemma 2.3 and Lemma 2.4 (the Mayer-Vietoris sequence) as well as the above isomorphism give that
$$ i^*: h^\bu_\TT(Y, \nabla^\xi, H)\cong h^\bu_\TT(N, \nabla^\xi, H)\cong h^\bu_\TT(F, i^*\nabla^\xi, i^*H)$$
is an isomorphism.

\section{Twisted Bismut-Chern character}

In \cite{BM00}, it was argued that D-brane charges in a background $H$-flux take values in twisted
K-theory of spacetime $Z$, $K^\bullet(Z,H)$. The Chern-Weil representatives of the twisted Chern character
$Ch_H:  K^\bullet(Z,H) \to H^\bullet(Z,H)$ were defined and its properties studied in \cite{BCMMS}.

Our goal in this section is to show that there is a
refinement of the twisted Chern character $Ch_H$ to the twisted Bismut-Chern
character $BCh_H: K^\bullet(Z,H) \to h_\TT^\bullet(LZ, \cL^B; \bar H)$ having the property that the following diagram commutes,
\begin{equation*}
\xymatrix{K^\bullet(Z,H)\ar[dr]_{Ch_H}\ar[rr]^{BCh_H}&&
h_\TT^\bullet(LZ, \cL^B; \bar H)\ar[dl]^{res}\\
&H(\Omega^\bu(Z)[[u, u^{-1}]], d+u^{-1}H) &}
\end{equation*}

\smallskip
When $H=0$, this reduces to Bismut's original construction in \cite{B85}.

We begin by reviewing geometric representatives of twisted K-theory and the Chern-Weil representative of the twisted Chern character.

\subsection{Gerbe modules with connections} Let $\{U_\alpha\}$ be a Brylinski cover of $Z$ and $E=\{E_{\alpha}\}$
be a collection of (infinite dimensional) Hilbert bundles $E_{\alpha}\to U_{\alpha}$ whose structure group is reduced to
$U_{\gI}$, which are unitary operators on the model Hilbert space $\gH$ of the form identity + trace class operator.
Here $\gI$ denotes the Lie algebra of trace class operators on $\gH$.
In addition, assume that on the overlaps $U_{\alpha\beta}$ that
there are isomorphisms
\beq
\phi_{\alpha\beta}: L_{\alpha\beta} \otimes E_\beta \cong E_\alpha,
\eeq
which are consistently defined on
triple overlaps because of the gerbe property. Then $\{E_{\alpha}\}$ is said to be a {\em gerbe module} for the gerbe
$\{L_{\alpha\beta}\}$.

A {\em gerbe module connection} $\nabla^E$ is a collection of connections $\{\nabla^E_{\alpha}\}$ is of the form $\nabla^E_{\alpha} = d + A_\alpha^E$ where $A_\alpha^E
\in \Omega^1(U_\alpha)\otimes \gI$ whose curvature $F^E_\alpha$ on the overlaps $U_{\alpha\beta}$ satisfies
\beq
\phi_{\alpha\beta}^{-1}(F^E_\alpha) \phi_{\alpha\beta} =  F^L_{\alpha\beta} I  +    F^E_\beta
\eeq
Using equation \eqref{gerbeconn}, this becomes
\beq
\phi_{\alpha\beta}^{-1}( B_\alpha I + F^E_\alpha ) \phi_{\alpha\beta} = B_{\beta} I  +    F^E_\beta.
\eeq
It follows that $\exp(-B)\Tr\left(\exp(-F^E) - I\right)$ is a globally well defined differential form on $Z$
of even degree. Notice that $\Tr(I)=\infty$ which is why we need to consider the subtraction.

\subsection{Geometric representatives of twisted K-theory and the twisted Chern character} Let $E=\{E_{\alpha}\}$ and $E'=\{E'_{\alpha}\}$ 
be a {gerbe modules} for the gerbe $\{L_{\alpha\beta}\}$. Then an element of twisted K-theory $K^0(Z, H)$
is represented by the pair $(E, E')$, see \cite{BCMMS}. Two such pairs $(E, E')$ and $(G, G')$ are equivalent
if $E\oplus G' \oplus K \cong E' \oplus G \oplus K$ as gerbe modules for some gerbe module $K$ for the gerbe $\{L_{\alpha\beta}\}$.
We can assume without loss of generality that these gerbe modules $E, E'$ are modeled on the same Hilbert space
$\gH$, after a choice of isomorphism if necessary.

Suppose that $\nabla^E, \nabla^{E'}$ are gerbe module connections on the gerbe modules $E, E'$ respectively. Then we can define the {\em twisted Chern character} as
\begin{align*}
Ch_H &: K^0(Z, H) \to H^{even}(Z, H)\\
Ch_H(E, E')&= \exp(-B)\Tr\left(\exp(-F^E) - \exp(-F^{E'})\right)
\end{align*}
That this is a well defined homomorphism is explained in \cite{BCMMS, MS}. To define the twisted Chern character landing in $\left(\Omega^\bu(Z)[[u, u^{-1}]]\right)_{(d+u^{-1}H)-cl}$, simply replace the above formula by
$$Ch_H(E, E')= \exp(-u^{-1}B)\Tr\left(\exp(-u^{-1}F^E) - \exp(-u^{-1}F^{E'})\right).$$

\subsection{Defining the twisted Bismut-Chern character} For a bundle $\xi$ with connection, denote the parallel transport of $\xi$ on a loop $\gamma$ from $\gamma_0$ to $\gamma_s$ by ${}^\xi\tau^0_s$ and ${}^\xi\tau^s_0$ is defined by $({}^\xi\tau^0_s)^{-1}$.

Let $\omega_\alpha \in \Omega^\bu(LU_\alpha)^\TT[[u, u^{-1}]]$ be defined by (the invariance can be seen from (49)) 
\h
\begin{split}
\omega_\alpha=&\mathrm{Tr}\left[ \left(I+\sum_{n=1}^\infty {(-u)}^{-n}\!\!\!\!\!\!\!\!\! \int \limits_{0\leq s_1\leq \cdots\leq s_n\leq 1}\!\!\!\!\!\!\!\!\!{}^{E_\alpha}\tau^{s_1}_0(\widehat{F^E_\alpha}_{s_1})\circ \cdots\circ {}^{E_\alpha}\tau^{s_n}_0(\widehat{F^E_\alpha }_{s_n})\right)\circ {}^{E_\alpha}\tau^1_0 \right.\\
&\left.\ \ \ \ \ - \left(I+\sum_{n=1}^\infty {(-u)}^{-n} \!\!\!\!\!\!\!\!\! \int \limits_{0\leq s_1\leq \cdots\leq s_n\leq 1}\!\!\!\!\!\!\!\!\!{}^{E_\alpha'}\tau^{s_1}_0(\widehat{F^{E'}_\alpha }_{s_1})\circ \cdots\circ {}^{E_\alpha'}\tau^{s_n}_0(\widehat{F^{E'}_\alpha }_{s_n})\right)\circ {}^{E_\alpha'}\tau^1_0  \right].
\end{split}
\e
Define an element $BCh_{H, \alpha}(\nabla^E, \nabla^{E'})\in \Omega^\bu(LU_\alpha, \cL^B)^\TT[[u, u^{-1}]]$ (the invariance can be seen from Theorem 3.1) by
\be \label{localBCh}
\begin{split}
&BCh_{H, \alpha}(\nabla^E, \nabla^{E'})\\
=&\left(1+\sum_{n=1}^\infty {(-u)}^{-n}\int\limits_{0\leq s_1\leq \cdots\leq s_n\leq 1}\widehat{B_\alpha }_{s_1}\cdots\widehat{B_\alpha }_{s_n}\right)\cdot\omega_\alpha\otimes\sigma_\alpha\\
=&\left(1+\sum_{n=1}^\infty {(-u)}^{-n}\int\limits_{0\leq s_1\leq \cdots\leq s_n\leq 1}\widehat{B_\alpha }_{s_1}\cdots\widehat{B_\alpha }_{s_n}\right)\\
&\cdot \mathrm{Tr}\left[ \left(I+\sum_{n=1}^\infty {(-u)}^{-n}\int\limits_{0\leq s_1\leq \cdots\leq s_n\leq 1}{}^{E_\alpha}\tau^{s_1}_0(\widehat{F^E_\alpha}_{s_1})\circ \cdots\circ {}^{E_\alpha}\tau^{s_n}_0(\widehat{F^E_\alpha }_{s_n})\right)\circ {}^{E_\alpha}\tau^1_0 \right.\\
&\left.\ \ \ \ \ - \left(I+\sum_{n=1}^\infty {(-u)}^{-n}\int\limits_{0\leq s_1\leq \cdots\leq s_n\leq 1}{}^{E_\alpha'}\tau^{s_1}_0(\widehat{F^{E'}_\alpha }_{s_1})\circ \cdots\circ {}^{E_\alpha'}\tau^{s_n}_0(\widehat{F^{E'}_\alpha }_{s_n})\right)\circ {}^{E_\alpha'}\tau^1_0  \right]\otimes \sigma_\alpha.
\end{split}
\ee

In the above and the following, to save space, we drop the $ds_1ds_2\cdots ds_n$ in the integration. Note that the subtraction of the terms coming from $E$ and $E'$ is essential here to take trace.

The $\{BCh_{H, \alpha} \}$ patch together to be a global form in $\Omega^\bu(LZ, \cL^B)^\TT [[u, u^{-1}]]$.
Actually
\be \label{BChMatch}
\begin{split}
&BCh_{H, \beta}(\nabla^E, \nabla^{E'})\\
=&\left(1+\sum_{n=1}^\infty {(-u)}^{-n}\int_{0\leq s_1\leq \cdots\leq s_n\leq 1}\widehat{B_\beta }_{s_1}\cdots\widehat{B_\beta }_{s_n}\right)\\
&\cdot \mathrm{Tr}\left[ \left(I+\sum_{n=1}^\infty {(-u)}^{-n} \int\limits_{0\leq s_1\leq \cdots\leq s_n\leq 1}{}^{E_\beta}\tau^{s_1}_0(\widehat{F^E_\beta }_{s_1})\circ \cdots\circ {}^{E_\beta}\tau^{s_n}_0(\widehat{F^E_\beta }_{s_n})\right)\circ {}^{E_\beta}\tau^1_0 \right.\\
&\left.\ \ \ \ \ - \left(I+\sum_{n=1}^\infty {(-u)}^{-n}\int\limits_{0\leq s_1\leq \cdots\leq s_n\leq 1}{}^{E_\beta'}\tau^{s_1}_0(\widehat{F^{E'}_\beta }_{s_1})\circ \cdots\circ {}^{E_\beta'}\tau^{s_n}_0(\widehat{F^{E'}_\beta }_{s_n})\right)\circ {}^{E_\beta'}\tau^1_0  \right]\otimes \sigma_\beta\\
=&\left(1+\sum_{n=1}^\infty {(-u)}^{-n}\int_{0\leq s_1\leq \cdots\leq s_n\leq 1}(\widehat{B_\alpha}_{s_1}+\widehat{F^L_{\alpha\beta}}_{s_1})\cdots
(\widehat{B_\alpha}_{s_n}+\widehat{F^L_{\alpha\beta}}_{s_n})\right)\\
&\cdot \mathrm{Tr}\left[ \left(I+\sum_{n=1}^\infty {(-u)}^{-n}\int\limits_{0\leq s_1\leq \cdots\leq s_n\leq 1}{}^{E_\beta}\tau^{s_1}_0(\widehat{F^E_\beta }_{s_1})\circ \cdots\circ {}^{E_\beta}\tau^{s_n}_0(\widehat{F^E_\beta }_{s_n})\right)\circ {}^{E_\beta}\tau^1_0 \right.\\
&\left.\ \ \ \ \ \ \ - \left(I+\sum_{n=1}^\infty {(-u)}^{-n}\int\limits_{0\leq s_1\leq \cdots\leq s_n\leq 1}{}^{E_\beta'}\tau^{s_1}_0(\widehat{F^{E'}_\beta }_{s_1})\circ \cdots\circ {}^{E_\beta'}\tau^{s_n}_0(\widehat{F^{E'}_\beta }_{s_n})\right)\circ {}^{E_\beta'}\tau^1_0  \right]\\
&\ \ \ \ \ \ \otimes (e^{\int\limits_0^1i_K A_{\alpha_\beta}})\sigma_\alpha\\
=&\left(1+ \sum_{n=1}^\infty {(-u)}^{-n}\int_{0\leq s_1\leq \cdots\leq s_n\leq 1}\widehat{B_\alpha }_{s_1}\cdots\widehat{B_\alpha}_{s_n}\right)\\
&\ \ \ \cdot \left(1+\sum_{n=1}^\infty {(-u)}^{-n}\int_{0\leq s_1\leq \cdots\leq s_n\leq 1}\widehat{F^L_{\alpha\beta}}_{s_1}
\cdots
\widehat{F^L_{\alpha\beta}}_{s_n}\right) {}^{L_{\alpha\beta}}\tau^1_0\\
&\cdot \mathrm{Tr}\left[ \left(I+ \sum_{n=1}^\infty {(-u)}^{-n}\int\limits_{0\leq s_1\leq \cdots\leq s_n\leq 1}{}^{E_\beta}\tau^{s_1}_0(\widehat{F^E_\beta }_{s_1})\circ \cdots\circ {}^{E_\beta}\tau^{s_n}_0(\widehat{F^E_\beta }_{s_n})\right)\circ {}^{E_\beta}\tau^1_0 \right.\\
&\left.\ \ \ \ \  - \left(I+\sum_{n=1}^\infty {(-u)}^{-n}\int\limits_{0\leq s_1\leq \cdots\leq s_n\leq 1}{}^{E_\beta'}\tau^{s_1}_0(\widehat{F^{E'}_\beta }_{s_1})\circ \cdots\circ {}^{E_\beta'}\tau^{s_n}_0(\widehat{F^{E'}_\beta }_{s_n})\right)\circ {}^{E_\beta'}\tau^1_0  \right]\otimes\sigma_\alpha\\
\end{split}
\ee

So we have
\h
\begin{split}
&BCh_{H, \beta}(\nabla^E, \nabla^{E'})\\
=&\left(1+\sum_{n=1}^\infty {(-u)}^{-n}\int\limits_{0\leq s_1\leq \cdots\leq s_n\leq 1}\widehat{B_\alpha }_{s_1}\cdots\widehat{B_\alpha}_{s_n}\right)\\
&\cdot \mathrm{Tr}\left[ \!\! \left(\!\!I\!+\!\sum_{n=1}^\infty {(-u)}^{-n}\!\!\!\!\!\!\!\!\int\limits_{0\leq s_1\leq \cdots\leq s_n\leq 1}\!\!\!\!\!\!\!\!{}^{L_{\alpha\beta}\otimes E_\beta}\tau^{s_1}_0(\widehat{F^L_{\alpha\beta}I\!+\!F^E_\beta }_{s_1})\!\circ\! \cdots\!\circ\! {}^{L_{\alpha\beta}\otimes E_\beta}\tau^{s_n}_0(\widehat{F^L_{\alpha\beta}I\!+\!F^E_\beta }_{s_n})\right)\!\circ\! {}^{L_{\alpha\beta}\otimes E_\beta}\tau^1_0 \right.\\
&\left.-\!\! \left(\!\!I\!+\!\sum_{n=1}^\infty {(-u)}^{-n}\!\!\!\!\!\!\!\!\int\limits_{0\leq s_1\leq \cdots\leq s_n\leq 1}\!\!\!\!\!\!\!\!{}^{L_{\alpha\beta}\otimes E'_\beta}\tau^{s_1}_0(\widehat{F^L_{\alpha\beta}I\!+\!F^{E'_\beta} }_{s_1})\!\circ\! \cdots\circ {}^{L_{\alpha\beta}\otimes E'_\beta}\tau^{s_n}_0(\widehat{F^L_{\alpha\beta}I\!+\!F^{E'_\beta} }_{s_n})\right)\!\!\circ\!\! {}^{L_{\alpha\beta}\otimes E'_\beta}\tau^1_0 \right]\otimes\sigma_\alpha\\
=&\left(1+\sum_{n=1}^\infty {(-u)}^{-n}\int\limits_{0\leq s_1\leq \cdots\leq s_n\leq 1}\widehat{B_\alpha }_{s_1}\cdots\widehat{B_\alpha }_{s_n}\right)\\
&\cdot \mathrm{Tr}\left[ \left(I+\sum_{n=1}^\infty {(-u)}^{-n}\int\limits_{0\leq s_1\leq \cdots\leq s_n\leq 1}{}^{E_\alpha}\tau^{s_1}_0(\widehat{F^E_\alpha}_{s_1})\circ\cdots\circ {}^{E_\alpha}\tau^{s_n}_0(\widehat{F^E_\alpha }_{s_n})\right)\circ {}^{E_\alpha}\tau^1_0 \right.\\
&\left.\ \ \ \ \ \ \ - \left(I+\sum_{n=1}^\infty {(-u)}^{-n}\int\limits_{0\leq s_1\leq \cdots\leq s_n\leq 1}{}^{E_\alpha'}\tau^{s_1}_0(\widehat{F^{E'}_\alpha }_{s_1})\circ \cdots\circ {}^{E_\alpha'}\tau^{s_n}_0(\widehat{F^{E'}_\alpha }_{s_n})\right)\circ {}^{E_\alpha'}\tau^1_0  \right]\otimes \sigma_\alpha.\\
=&BCh_{H, \alpha}(\nabla^E, \nabla^{E'})\\
\end{split}
\e

Define the {\it twisted Bismut-Chern character form} $BCh_H(\nabla^E, \nabla^{E'})\in \Omega^\bu(LZ, \cL^B)^\TT[[u, u^{-1}]]$ to be the global form patched together from the forms construction as (\ref{localBCh}). It is easily seen that when restricted to constant loops, the twisted Bismut-Chern character form degenerates to the twisted Chern character form and hence the commutativity of (0.1) follows.

\begin{theorem} (i) We have $(\nabla^{\cL^B} -u i_K +u^{-1} \bar H)BCh_H(\nabla^E, \nabla^{E'})=0;$\newline
(ii) The exotic twisted $\TT$-equivariant cohomology class $[BCh_H(\nabla^E, \nabla^{E'})]$ does not depend on the choice of connections $\nabla^E, \nabla^{E'}$.

\end{theorem}

\begin{proof}(i) To prove
$$(\nabla^{\cL^B} -u i_K +u^{-1} \bar H)BCh_H(\nabla^E, \nabla^{E'})=0,$$
we only have to show that
\be  \label{localclosed} (d-i_K \bar B_\alpha-ui_K+u^{-1}\bar H)\left(\left(1+\sum_{n=1}^\infty {(-u)}^{-n}\int\limits_{0\leq s_1\leq s_2\leq \cdots\leq s_n\leq 1}\widehat{B_\alpha }_{s_1}\widehat{B_\alpha }_{s_2}\cdots\widehat{B_\alpha }_{s_n}\right)\cdot\omega_\alpha\right)=0. \ee

From Bismut's result in \cite{B85}, we know that
\be 
(d-ui_K)\omega_\alpha=0.
\ee

Moreover
\be \label{d-iu_K}
\begin{split}
&(d-ui_K)\left(1+ \sum_{n=1}^\infty{(-u)}^{-n}\int\limits_{0\leq s_1\leq \cdots\leq s_n\leq 1}\widehat{B_\alpha }_{s_1}\cdots\widehat{B_\alpha }_{s_n}\right)\\
=&\sum_{n=1}^\infty{(-u)}^{-n}\sum\limits_{i=1}^n\int\limits_{0\leq s_1\leq \cdots\leq s_n\leq 1}\widehat{B_\alpha }_{s_1}\cdots\widehat{dB_\alpha }_{s_i}\cdots\widehat{B_\alpha }_{s_n}\\
&+\sum_{n=1}^\infty (-u)^{-n+1}\sum\limits_{i=1}^n\int\limits_{0\leq s_1\leq \cdots\leq s_n\leq 1}\widehat{B_\alpha }_{s_1}\cdots\widehat{i_K B_\alpha }_{s_i}\cdots\widehat{B_\alpha }_{s_n}\\
=&\sum_{n=1}^\infty{(-u)}^{-n}\sum\limits_{i=1}^n\int\limits_{0\leq s_1\leq \cdots\leq s_n\leq 1}\widehat{B_\alpha }_{s_1}\cdots\widehat{H}_{s_i}\cdots\widehat{B_\alpha }_{s_n}\\
&+\sum_{n=1}^\infty (-u)^{-n+1}\sum\limits_{i=1}^n\int\limits_{0\leq s_1\leq \cdots\leq s_n\leq 1}\widehat{B_\alpha }_{s_1}\cdots\widehat{i_K B_\alpha }_{s_i}\cdots\widehat{B_\alpha }_{s_n}\\
=&\left(-u^{-1}\int_0^1\widehat{H}_{s}ds\right) \left(1+ \sum_{n=1}^\infty{(-u)}^{-n}\int\limits_{0\leq s_1\leq \cdots\leq s_n\leq 1}\widehat{B_\alpha }_{s_1}\widehat{B_\alpha }_{s_2}\cdots\widehat{B_\alpha }_{s_n}\right)\\
&+\left(\int_0^1\widehat{i_K B_\alpha}_{s}ds\right) \left(1+ \sum_{n=1}^\infty{(-u)}^{-n}\int\limits_{0\leq s_1\leq \cdots\leq s_n\leq 1}\widehat{B_\alpha }_{s_1}\cdots\widehat{B_\alpha }_{s_n}\right).\\
\end{split}
\ee

So
\be
\begin{split}
\label{dofHB}
&(d-ui_K)\left(1+ \sum_{n=1}^\infty{(-u)}^{-n}\int\limits_{0\leq s_1\leq \cdots\leq s_n\leq 1}\widehat{B_\alpha }_{s_1}\cdots\widehat{B_\alpha }_{s_n}\right)\\
=&(i_K \bar B-u^{-1}\bar H)\left(1+ \sum_{n=1}^\infty{(-u)}^{-n}\int\limits_{0\leq s_1\leq \cdots\leq s_n\leq 1}\widehat{B_\alpha }_{s_1}\cdots\widehat{B_\alpha }_{s_n}\right).
\end{split}
\ee
Therefore it is clear that (\ref{localclosed}) holds.

$\, $

\noindent (ii) Now suppose we have two pairs of gerbe module connections $(\nabla_0^E, \nabla_0^{E'})$ and $(\nabla_1^E, \nabla_1^{E'})$.
Let $\nabla_t^E=(1-t)\nabla_0^E+t\nabla_1^E$ and $\nabla_t^{E'}=(1-t)\nabla_0^{'E}+t\nabla_1^{E'}$.  Denote the curvatures of these two connections by $F^{E,t}$ and $F^{E', t}$.

Let $A^E=\nabla_1^E-\nabla_0^E$ and $A^{E'}=\nabla_1^{E'}-\nabla_0^{E'}$.  Actually $A^E=\{A^E_\alpha\}$ with the relation $\phi_{\alpha\beta}^{-1}(A^E_\alpha) \phi_{\alpha\beta} =  A^E_\beta$.
Similar relations hold for $\{A^{E'}_\alpha\}$.

Let $\eta_\alpha \in \Omega^\bu(LU_\alpha)^\TT[[u, u^{-1}]]$ be defined by (the invariance can be seen from (55))
$$ $$
$$ $$
\h
\eta_\alpha =\mathrm{Tr}\left[ -u^{-1}\! \left(\int_0^1{}^{{E^t_\alpha}}\tau^s_0(\widehat{A^E_\alpha}_s) ds\!\circ\!\! \left(\! \!  I+\! \sum_{n=1}^\infty {(-u)}^{-n}\!\! \! \! \! \! \! \! \!\! \! \! \!  \!\int\limits_{0\leq s_1\leq \cdots\leq s_n\leq 1}\!\!\!\!\!\!\!\!{}^{E^t_\alpha}\tau^{s_1}_0(\widehat{F^{E,t}_{\alpha}}_{s_1})\! \circ \! \cdots\circ {}^{E^t_\alpha}\tau^{s_n}_0(\widehat{F^{E,t}_\alpha}_{s_n})\right)\right)\!\circ\! {}^{{E^t_\alpha}}\tau^1_0 \right.
\e
\h
\left.\!\! \!+u^{-1}\!\! \left(\! \! \int_0^1{}^{E'^t_\alpha}\tau^s_0(\widehat{A^{E'}_\alpha}_s) ds \!\circ\!\! \left(\! \! I\!+\!\sum_{n=1}^\infty {(-u)}^{-n}\! \! \! \!\! \! \! \! \!\! \! \! \! \! \int\limits_{0\leq s_1\leq \cdots\leq s_n\leq 1}\!\!\!\!\!\!\!\!{}^{E'^t_\alpha}\tau^{s_1}_0(\widehat{F^{E',t}_\alpha}_{s_1})\!\circ\! \cdots\circ {}^{E'^t_\alpha}\tau^{s_n}_0(\widehat{F^{E',t}_\alpha }_{s_n})\right)\right)\!\circ\! {}^{E'^t_\alpha}\tau^1_0 \right].
\e

Define an element $BCS_{H, \alpha}(\nabla_0^E, \nabla_0^{E'}; \nabla_1^E, \nabla_1^{E'})\in \Omega^\bu(LU_\alpha, \cL^B)^\TT[[u, u^{-1}]]$ by

\be \label{localBCS}
\begin{split}
&BCS_{H, \alpha}(\nabla_0^E, \nabla_0^{E'}; \nabla_1^E, \nabla_1^{E'})\\
=&\left(1+\sum_{n=1}^\infty {(-u)}^{-n}\int\limits_{0\leq s_1\leq \cdots\leq s_n\leq 1}\widehat{B_\alpha }_{s_1}\cdots\widehat{B_\alpha }_{s_n}\right)\cdot\eta_\alpha\otimes\sigma_\alpha\\
=&\left(1+\sum_{n=1}^\infty {(-u)}^{-n}\int\limits_{0\leq s_1\leq \cdots\leq s_n\leq 1}\widehat{B_\alpha }_{s_1}\cdots\widehat{B_\alpha }_{s_n}\right)\\
\end{split}
\ee

\h
\cdot  \! \mathrm{Tr}\left[ -u^{-1}\! \left(\int_0^1{}^{{E^t_\alpha}}\tau^s_0(\widehat{A^E_\alpha}_s) ds\!\circ\!\! \left(\! \!  I+\! \sum_{n=1}^\infty {(-u)}^{-n}\!\! \! \! \! \! \! \! \int\limits_{0\leq s_1\leq \cdots\leq s_n\leq 1}\!\!\!\!\!\!\!\!{}^{E^t_\alpha}\tau^{s_1}_0(\widehat{F^{E,t}_{\alpha}}_{s_1})\! \circ \! \cdots\circ {}^{E^t_\alpha}\tau^{s_n}_0(\widehat{F^{E,t}_\alpha}_{s_n})\right)\right)\!\circ\! {}^{{E^t_\alpha}}\tau^1_0 \right.
\e
\h
\left.\!\! \!+u^{-1}\!\! \left(\! \! \int_0^1{}^{E'^t_\alpha}\tau^s_0(\widehat{A^{E'}_\alpha}_s) ds \!\circ\!\! \left(\! \! I\!+\!\sum_{n=1}^\infty {(-u)}^{-n}\! \! \! \!\! \! \! \! \int\limits_{0\leq s_1\leq \cdots\leq s_n\leq 1}\!\!\!\!\!\!\!\!{}^{E'^t_\alpha}\tau^{s_1}_0(\widehat{F^{E',t}_\alpha}_{s_1})\!\circ\! \cdots\circ {}^{E'^t_\alpha}\tau^{s_n}_0(\widehat{F^{E',t}_\alpha }_{s_n})\right)\right)\!\circ\! {}^{E'^t_\alpha}\tau^1_0 \right]\otimes \sigma_\alpha.
\e

Similar to what we did in {\ref{BChMatch}}, we can show that $\{BCS_{H, \alpha}(\nabla_0^E, \nabla_0^{E'}; \nabla_1^E, \nabla_1^{E'})\}$ patch together to give us a global form $BCS_{H}(\nabla_0^E, \nabla_0^{E'}; \nabla_1^E, \nabla_1^{E'})\in \Omega^\bu(LZ, \cL^B)^\TT[[u, u^{-1}]].$  We call this form the {\it twisted Bismut-Chern-Simons transgression term}.

From Bismut's result in \cite{B85}, we know that
\be
\omega_\alpha(\nabla_1^E, \nabla_1^{E'})-\omega_\alpha(\nabla_0^E, \nabla_0^{E'})=(d-ui_K)\eta_\alpha.
\ee
Combining (\ref{dofHB}), we see that
\be
\begin{split}
&\left(1+\sum_{n=1}^\infty {(-u)}^{-n}\int\limits_{0\leq s_1\leq \cdots\leq s_n\leq 1}\widehat{B_\alpha }_{s_1}\cdots\widehat{B_\alpha }_{s_n}\right)\cdot\omega_\alpha(\nabla_1^E, \nabla_1^{E'})\\
-&\left(1+\sum_{n=1}^\infty {(-u)}^{-n}\int\limits_{0\leq s_1\leq \cdots\leq s_n\leq 1}\widehat{B_\alpha }_{s_1}\cdots\widehat{B_\alpha }_{s_n}\right)\cdot \omega_\alpha(\nabla_0^E, \nabla_0^{E'})\\
=&(d-i_K \bar B_\alpha-ui_K+u^{-1}\bar H)\left( \left(1+\sum_{n=1}^\infty {(-u)}^{-n}\int\limits_{0\leq s_1\leq \cdots\leq s_n\leq 1}\widehat{B_\alpha }_{s_1}\cdots\widehat{B_\alpha }_{s_n}\right)\cdot\eta_\alpha\right).
\end{split}
\ee

Therefore, we have
\be
\begin{split}
&BCh_H(\nabla_1^E, \nabla_1^{E'})-BCh_H(\nabla_0^E, \nabla_0^{E'})\\
=&(\nabla^{\cL^B} -u i_K +u^{-1} \bar H)BCS_{H}(\nabla_0^E, \nabla_0^{E'}; \nabla_1^E, \nabla_1^{E'}).
\end{split}
\ee

\end{proof}

If we change the curving of the gerbe from $\{B_\alpha\}$ to $\{B_\alpha+Q\}$, where $Q\in \Omega^2(Z)$, then the holonomy line bundle $\cL^B$ is unchanged while the connection $\nabla^{\cL^B}$ is change to $\nabla^{\cL^B}-\tau(Q)$ and  the 3-curvature of the gerbe is changed to $H+dQ$. So the differential $\nabla^{\cL^B} -u i_K +u^{-1} \bar H$ is changed to $\nabla^{\cL^B} -\tau(Q)- ui_K + u^{-1}\overline{H+dQ}.$

Consider the map 
\be \gV_Q: \Omega^\bullet(LZ, \cL^B)^\TT[[u, u^{-1}]]\to \Omega^\bullet(LZ, \cL^B)^\TT[[u, u^{-1}]]\ee such that 
$$ \gV_Q(\omega)=\left(\!1\!+\!\sum_{n=1}^\infty {(-u)}^{-n}\!\!\!\!\!\!\!\!\!\!\!\!\int\limits_{0\leq s_1\leq \cdots\leq s_n\leq 1}\!\!\!\!\!\!\widehat{Q }_{s_1}\cdots\widehat{Q }_{s_n}\right)\cdot \omega.$$

Suppose $(\nabla^{\cL^B} -u i_K +u^{-1} \bar H)\omega=0$, then (similar to (\ref{d-iu_K}))
\be
\begin{split}
&(\nabla^{\cL^B} -\tau(Q)- ui_K + u^{-1}\overline{H+dQ})\left(\left(\!1\!+\!\sum_{n=1}^\infty {(-u)}^{-n}\!\!\!\!\!\!\!\!\!\!\!\!\int\limits_{0\leq s_1\leq \cdots\leq s_n\leq 1}\!\!\!\!\!\!\widehat{Q }_{s_1}\cdots\widehat{Q }_{s_n}\right)\cdot \omega\right)\\
=&(-\tau(Q) + u^{-1}\overline{dQ})\left(\!1\!+\!\sum_{n=1}^\infty {(-u)}^{-n}\!\!\!\!\!\!\!\!\!\!\!\!\int\limits_{0\leq s_1\leq \cdots\leq s_n\leq 1}\!\!\!\!\!\!\widehat{Q }_{s_1}\cdots\widehat{Q }_{s_n}\right)\cdot \omega\\
&+\left((d-ui_K)\left(\!1\!+\!\sum_{n=1}^\infty {(-u)}^{-n}\!\!\!\!\!\!\!\!\!\!\!\!\int\limits_{0\leq s_1\leq \cdots\leq s_n\leq 1}\!\!\!\!\!\!\widehat{Q }_{s_1}\cdots\widehat{Q }_{s_n}\right)\right)\cdot \omega\\
=&0.
\end{split}
\ee

This shows that we have an isomorphism (still denoted by $\gV_Q$)
$$\gV_Q:  h^\bu_\TT(LZ, \nabla^{\cL^B}: \bar H)\to  h^\bu_\TT(LZ, (\nabla^{\cL^B}-\tau(Q)): \overline {H+dQ}).$$

It is clear that when restricted to constant loops, we get an isomorphism
$$ V_Q: H(\Omega^\bu(Z)[[u, u^{-1}]], d+u^{-1}H) \to H(\Omega^\bu(Z)[[u, u^{-1}]], d+u^{-1}(H+dQ)),$$
sending $\eta$ to $ e^{-u^{-1}Q}\cdot\eta.$ 

Moreover, we see that 
\be 
\begin{split}
&BCh_{H+dQ, \alpha}(\nabla^E, \nabla^{E'})\\
=&\left(1+\sum_{n=1}^\infty {(-u)}^{-n}\int\limits_{0\leq s_1\leq \cdots\leq s_n\leq 1}\widehat{B_\alpha+Q }_{s_1}\cdots\widehat{B_\alpha +Q}_{s_n}\right)\\
&\cdot \mathrm{Tr}\left[ \left(I+\sum_{n=1}^\infty {(-u)}^{-n}\int\limits_{0\leq s_1\leq \cdots\leq s_n\leq 1}{}^{E_\alpha}\tau^{s_1}_0(\widehat{F^E_\alpha}_{s_1})\circ \cdots\circ {}^{E_\alpha}\tau^{s_n}_0(\widehat{F^E_\alpha }_{s_n})\right)\circ {}^{E_\alpha}\tau^1_0 \right.\\
&\left.\ \ \ \ \ - \left(I+\sum_{n=1}^\infty {(-u)}^{-n}\int\limits_{0\leq s_1\leq \cdots\leq s_n\leq 1}{}^{E_\alpha'}\tau^{s_1}_0(\widehat{F^{E'}_\alpha }_{s_1})\circ \cdots\circ {}^{E_\alpha'}\tau^{s_n}_0(\widehat{F^{E'}_\alpha }_{s_n})\right)\circ {}^{E_\alpha'}\tau^1_0  \right]\otimes \sigma_\alpha\\
=&\left(\!1\!+\!\sum_{n=1}^\infty {(-u)}^{-n}\!\!\!\!\!\!\!\!\!\!\!\!\int\limits_{0\leq s_1\leq \cdots\leq s_n\leq 1}\!\!\!\!\!\!\widehat{Q }_{s_1}\cdots\widehat{Q }_{s_n}\right)\left(1+\sum_{n=1}^\infty {(-u)}^{-n}\int\limits_{0\leq s_1\leq \cdots\leq s_n\leq 1}\widehat{B_\alpha}_{s_1}\cdots\widehat{B_\alpha}_{s_n}\right)\\
&\cdot \mathrm{Tr}\left[ \left(I+\sum_{n=1}^\infty {(-u)}^{-n}\int\limits_{0\leq s_1\leq \cdots\leq s_n\leq 1}{}^{E_\alpha}\tau^{s_1}_0(\widehat{F^E_\alpha}_{s_1})\circ \cdots\circ {}^{E_\alpha}\tau^{s_n}_0(\widehat{F^E_\alpha }_{s_n})\right)\circ {}^{E_\alpha}\tau^1_0 \right.\\
&\left.\ \ \ \ \ - \left(I+\sum_{n=1}^\infty {(-u)}^{-n}\int\limits_{0\leq s_1\leq \cdots\leq s_n\leq 1}{}^{E_\alpha'}\tau^{s_1}_0(\widehat{F^{E'}_\alpha }_{s_1})\circ \cdots\circ {}^{E_\alpha'}\tau^{s_n}_0(\widehat{F^{E'}_\alpha }_{s_n})\right)\circ {}^{E_\alpha'}\tau^1_0  \right]\otimes \sigma_\alpha\\
=&\gV_Q(BCh_{H, \alpha}(\nabla^E, \nabla^{E'})).
\end{split}
\ee
So we have
\be  \gV_Q(BCh_{H}(\nabla^E, \nabla^{E'}))=BCh_{H+dQ}(\nabla^E, \nabla^{E'}). \ee

It is also clear that
\be  V_Q(Ch_{H}(\nabla^E, \nabla^{E'}))=Ch_{H+dQ}(\nabla^E, \nabla^{E'}). \ee

Therefore we have the commutative diagram
\begin{equation}
\xymatrix 
@=4pc 
{ K^\bullet(Z, H) \ar@/_8pc/[dd]_{Ch_H}\ar[d]_{BCh_H}  \ar[r]^{\otimes \text{trivial gerbe with curving}\,\,\,  Q}
_{\text{and 3-curvature}\,\,\, dQ}
&
K^{\bullet }(Z, H+dQ)   \ar@/^8pc/[dd]^{Ch_{H+dQ}}\ar[d]^{BCh_{H+dQ}}   \\ h_\TT^\bullet(LZ, \nabla^{\cL^B}:\bar H)  \ar[r]^{\gV_Q}_{\cong} \ar[d]_{res}^\cong & 
h^\bu_\TT(LZ, (\nabla^{\cL^B}-\tau(Q)): \overline {H+dQ})\ar[d]^{res}_{\cong} \\
H^\bu(\Omega(Z)[[u, u^{-1}]], d+u^{-1}H) \ar[r]^{V_Q}_{\cong}& H(\Omega^\bu(Z)[[u, u^{-1}]], d+u^{-1}(H+dQ))
}
\end{equation}

\section{Relation to T-duality}
We apply our previous results to the study of T-duality. First we review the results
in  \cite{BEM04a, BEM04b}, where the following situation is studied.

\subsection{Review of T-duality for principal circle bundles in a background flux}

In \cite{BEM04a, BEM04b}, spacetime $Z$ was compactified in one direction.
More precisely, $Z$
is a principal $\bbT$-bundle over $X$

\begin{equation}\label{eqn:MVBx}
\begin{CD}
\bbT @>>> Z \\
&& @V\pi VV \\
&& X \end{CD}
\end{equation}
 classified up to isomorphism by its first Chern class
{ $c_1(Z)\in H^2(X,\ZZ)$}. Assume that spacetime $Z$ is endowed with an $H$-flux which is
a representative in the
degree 3 Deligne cohomology of $Z$, that is
$H\in\Omega^3(Z)$ with integral periods (for simplicity, we drop factors of $\frac{1}{2\pi i}$),
together with
the following data. Consider a local trivialization $U_\alpha \times \TT$ of $Z\to X$, where
$\{U_\alpha\}$ is a good cover of $X$. Let $H_\alpha = H\Big|_{ U_\alpha \times \TT}
= d B_\alpha$, where $B_\alpha \in \Omega^2(U_\alpha \times \TT)$ and finally, $B_\alpha -B_\beta = F_{\alpha\beta}
\in \Omega^1(U_{\alpha\beta} \times \TT)$.
 Then the choice of $H$-flux entails that we are given a local trivialization
 as above and locally defined 2-forms $B_\alpha$ on it, together with closed 2-forms $F_{\alpha\beta}$ defined on double overlaps,  that is, $(H, B_\alpha, F_{\alpha\beta})$. Also the first Chern class
 of $Z\to X$
  is represented in integral cohomology by $(F, A_\alpha)$ where
$\{A_\alpha\}$ is a connection 1-form on $Z\to X$ and $F = dA_\alpha$ is the curvature 2-form of $\{A_\alpha\}$.

The {  {T-dual}}  is another principal
$\bbT$-bundle over $M$, denoted by $\hat Z$,
  {}
\begin{equation}\label{eqn:MVBy}
\begin{CD}
\hat \bbT @>>> \hat Z \\
&& @V\hat \pi VV     \\
&& X \end{CD}
\end{equation}
To define it, we see that $\pi_* (H_\alpha) = d \pi_*(B_\alpha) = d {\hat A}_\alpha$,
 so that $\{{\hat A}_\alpha\}$ is a connection 1-form whose curvature $ d {\hat A}_\alpha = \hat F_\alpha =  \pi_*(H_\alpha)$
 that is, $\hat F = \pi_* H$. So let $\hat Z$ denote the principal
$\bbT$-bundle over $M$ whose first Chern class is  $\,\, c_1(\hat Z) = [\pi_* H, \pi_*(B_\alpha)] \in H^2(X; \ZZ) $.

The Gysin
sequence for $Z$ enables us to define a T-dual $H$-flux
$[\hat H]\in H^3(\hat Z,\ZZ)$, satisfying
\begin{equation} \label{eqn:MVBc}
c_1(Z) = \hat \pi_* \hat H \,,
\end{equation}
 where $\pi_* $
and similarly $\hat\pi_*$, denote the pushforward maps.
Note that $ \hat H$ is not fixed by this data, since any integer
degree 3 cohomology class on $X$ that is pulled back to $\hat Z$
also satisfies the requirements. However, $ \hat H$ is
determined uniquely (up to cohomology) upon imposing
the condition $[H]=[\hat H]$ on the correspondence space $Z\times_X \hat Z$
as will be explained now.

The {\em correspondence space} (sometimes called the doubled space) is defined as
$$
Z\times_X  \hat Z = \{(x, \hat x) \in Z \times \hat Z: \pi(x)=\hat\pi(\hat x)\}.
$$
Then we have the following commutative diagram,
\begin{equation*} \label{eqn:correspondence}
\xymatrix @=6pc @ur { (Z, [H]) \ar[d]_{\pi} &
(Z\times_X  \hat Z, [H]=[\hat H]) \ar[d]_{\hat p} \ar[l]^{p} \\ X & (\hat Z, [\hat H])\ar[l]^{\hat \pi}}
\end{equation*}
By requiring that
$$
p^*[H]={\hat p}^*[\hat H] \in H^3(Z\times_X  \hat Z, \ZZ),
$$
determines $[\hat H] \in H^3(  \hat Z, \ZZ)$  uniquely, via an application of the Gysin sequence.
An alternate way to see this is is explained below.

Let $(H, B_\alpha, F_{\alpha\beta}, L_{\alpha\beta})$ denote a gerbe with connection on $Z$.
We also choose a connection 1-form $A$ on $Z$.
Let $v$ denote the vectorfield generating the $S^1$-action on $Z$.
Then define $\widehat A_\alpha = -\imath_v B_\alpha$ on the chart $U_\alpha$ and
the connection 1-form $\widehat A= \widehat A_\alpha +d\widehat\theta_\alpha$
on the chart $U_\alpha\times  \widehat \TT$. In this way we get a T-dual circle bundle
$\widehat Z \to X$ with connection 1-form $\widehat A$.

Without loss of generality, we can assume that $H$ is $\TT$-invariant. Consider
$$
\Omega = H - A\wedge F_{\widehat A}
$$
where  $F_{\widehat A} = d {\widehat A}$ and $F_{A} = d {A}$ are the curvatures of $A$
and $\widehat A$ respectively. One checks that the contraction $i_v(\Omega)=0$ and
the Lie derivative $L_v(\Omega)=0$ so that $\Omega$ is a basic 3-form on $Z$, that is
$\Omega$ comes from the base $X$.

Setting
$$
\widehat H = F_A\wedge {\widehat A} + \Omega
$$
this defines the T-dual flux 3-form. One verifies that $\widehat H$ is a closed 3-form on $\widehat Z$.
It follows that on the correspondence space, one has as desired,
\begin{equation}
\widehat H = H + d (A\wedge \widehat A ).
\end{equation}

Our next goal is to determine the T-dual curving or B-field.
The Buscher rules imply that on the open sets $U_\alpha \times \TT\times \widehat \TT$ of the
correspondence space $Z\times_X \hat Z$, one has
\begin{equation}
\widehat B_\alpha = B_\alpha + A\wedge \widehat A - d\theta_\alpha \wedge d\widehat \theta _\alpha\,,
\end{equation}
Note that
\begin{equation}
\imath_v \widehat B_\alpha = \imath_v
\left( B_\alpha + A\wedge \widehat A - d\theta_\alpha \wedge d\widehat \theta _\alpha\right) =
-\widehat A_\alpha + \widehat A - d\widehat \theta_\alpha = 0
\end{equation}
so that $\widehat B_\alpha$ is indeed a 2-form on $\widehat Z$ and not just on the correspondence
space. Obviously, $d \widehat B_\alpha = \widehat H$. Following the descent equations one arrives at the complete
T-dual gerbe with connection, $(\widehat H, \widehat B_\alpha, \widehat F_{\alpha\beta}, \widehat L_{\alpha\beta})$.
cf. \cite{BMPR}.

The rules
for transforming the Ramond-Ramond (RR) fields can be encoded in the {\cite{BEM04a, BEM04b}} generalization of
{\em Hori's formula}
  {}
\begin{equation} \label{eqn:Hori}
T_*G =  \int_{\bbT} e^{ A \wedge \hat A }\ G \,,
\end{equation}
 where $G \in \Omega^\bullet(Z)^\bbT$ is the total RR fieldstrength,
\begin{center}
$G\in\Omega^{even}(Z)^\bbT \quad$ for {   { Type IIA}};\\
$G\in\Omega^{odd}(Z)^\bbT \quad$ for {   { Type IIB}},\\
\end{center}
and where the right hand side of equation \eqref{eqn:Hori} is an invariant differential form on $Z\times_X\hat Z$, and
the integration is along the $\bbT$-fiber of $Z$.


Recall that the twisted cohomology
is defined as the cohomology of the complex
$$H^\bullet(Z, H) = H^\bullet(\Omega^\bullet(Z), d_H=d+ H\wedge).$$
By the identity \eqref{eqn:Hori}, $T_*$ maps $d_H$-closed forms $G$ to $d_{\hat
H}$-closed forms $T_*G$.
 So T-duality $T_*$  induces a map on twisted cohomologies,
$$
T : H^\bullet(Z, H) \to H^{\bullet +1}(\hat Z, \hat H).
$$
Define the Riemannian metrics on $Z$ and $\hat Z$ respectively by
$$
g=\pi^*g_X+R^2\, A\odot A,\qquad \hat g=\hat\pi^*g_X+1/{R^2} \,\hat A\odot\hat A.
$$
where $g_X$ is a Riemannian metric on $X$.
 Then $g$ is $\TT$-invariant and the length of each circle fibre is $R$; $\hat g$
 is $\hat\TT$-invariant and the length of each circle fibre is $1/R$.

 The following theorem summarizes the main consequence of T-duality for principal circle bundles in a background flux.

\begin{theorem}[T-duality isomorphism \cite{BEM04a,BEM04b}]\label{thm:T-duality}
In the notation above, and
with the above choices of Riemannian metrics and flux forms, the map \eqref{eqn:Hori}
$$
T\colon\Omega^{\bar k}(Z)^\TT\to\Omega^{\overline{k+1}}(\hat Z)^{\hat\TT},
$$
for $k=0,1$, (where $\bar k$ denotes the parity of $k$) are isometries, inducing isomorphisms on twisted
cohomology groups,
\begin{equation}\label{T-duality-coh}
T : H^\bullet(Z, H) \stackrel{\cong}{\longrightarrow} H^{\bullet +1}(\hat Z, \hat H).
\end{equation}
Therefore under T-duality one has the exchange,
\begin{center}
{$R \Longleftrightarrow 1/R$}\quad and \quad
{{   {background H-flux} $\Longleftrightarrow$   {Chern class}}}
\end{center}
Moreover there is also an isomorphism of twisted K-theories,
\begin{equation}\label{T-duality-K}
T : K^\bullet(Z, H) \to K^{\bullet +1}(\hat Z, \hat H),
\end{equation}
such that the following diagram commutes,
\begin{equation}\label{T-duality-commute}
\xymatrix @=4pc 
{ K^\bullet(Z, H) \ar[d]_{Ch_H}  \ar[r]^{T} &
K^{\bullet +1}(\hat Z, \hat H)  \ar[d]^{Ch_{\hat H}}   \\ H^\bullet(Z, H)  \ar[r]_{T}& H^{\bullet +1}(\hat Z, \hat H)}
\end{equation}

\end{theorem}
The surprising {\em  { {new}}} phenomenon discovered in \cite{BEM04a,BEM04b} is that there is a
{\em   {  change in topology}} when the $H$-flux is non-trivial.

\subsection{T-duality: a loop space perspective}

As a consequence of our Localisation Theorem \ref{thm:localisation}, properties of the twisted Bismut-Chern 
character in section 3, and Theorem \ref{thm:T-duality}, 
we obtain a T-duality isomorphism of completed periodic exotic twisted $\TT$-equivariant cohomologies,

\begin{corollary}[T-duality and completed periodic exotic twisted $\TT$-equivariant cohomologies]  In the notation above, there is an isomorphism
\beq
T: h_\TT^\bullet(LZ, \nabla^{\cL^B}:\bar H) \stackrel{\cong}{\longrightarrow} h_\TT^{\bullet+1}(L{\widehat Z}, \nabla^{{\cL^{\widehat B}}}:\bar{\widehat H}),
\eeq
such that the following diagram commutes,
\begin{equation}\label{T-duality-commute2}
\xymatrix @=4pc 
{ K^\bullet(Z, H)  \ar@/_8pc/[dd]_{Ch_H}\ar[d]_{BCh_H}  \ar[r]^{T}_{\cong} &
K^{\bullet +1}(\hat Z, \hat H)  \ar@/^8pc/[dd]^{Ch_{\hat H}} \ar[d]^{BCh_{\hat H}}   \\ h_\TT^\bullet(LZ, \nabla^{\cL^B}:\bar H)  \ar[r]_{T} \ar[d]_{res}^\cong & 
h_\TT^{\bullet+1}(L{\widehat Z}, \nabla^{{\cL^{\widehat B}}}:\bar{\widehat H})\ar[d]^{res}_{\cong} \\
H^\bu(\Omega(Z)[[u, u^{-1}]], d+u^{-1}H) \ar[r]_{T}^{\cong}& H^{\bu+1}(\Omega(\widehat{Z})[[u, u^{-1}]], d+u^{-1}\widehat{H}) 
}
\end{equation}
\end{corollary}

To complete the proof of the corollary, we prove the following lemma
\begin{lemma}\label{lem:equiv}
In the notation above, $H^\bu(\Omega(Z)[[u, u^{-1}]], d+u^{-1}H) \cong H^\bu(Z, H)[[u, u^{-1}]]$.
\end{lemma}
\begin{proof}
Let $c_u : \Omega(Z)[[u, u^{-1}]] \to \Omega(Z)[[u, u^{-1}]]$ be defined by multiplication by $u^{-[\frac{i}{2}]}$
on $i$-forms. Then $c_u(H)= u^{-1} H$ and one checks that $c_u \circ (d+H) = u^{-k} (d + u^{-1}H) \circ c_u$ 
on $\Omega^{\bar k}(Z)[[u, u^{-1}]]$ for $k=0,1$, with $\bar k$ denoting the parity of $k$. Observing that 
$$H^\bu(\Omega(Z)[[u, u^{-1}]], d + H) = H^\bu(Z, H)[[u, u^{-1}]],$$ we see that
 $c_u$ gives the desired isomorphism.
\end{proof}
This isomorphism is  interpreted as {\em T-duality from a loop space perspective}, giving an equivalence (rationally) between 
D-brane charges in a background H-flux in type IIA and IIB string theories. As a consequence of the above, 
we propose that the configuration space 
of Ramond-Ramond fields to be the space of differential forms with coefficients in the holonomy line bundle 
 on loop space $LZ$ of the gerbe $\cG_B$ on spacetime $Z$, and that are closed under the equivariantly closed superconnection.

A refinement and extension of these results is the topic of the work in progress \cite{FHVM}.

\end{document}